\documentclass[twocolumn]{autart}

\usepackage{times,enumerate} 
\usepackage{amsmath} 
\usepackage{amssymb,adjustbox}
\usepackage{setspace,color}
\usepackage{graphicx}
\usepackage{epsfig} 
\usepackage{subfig}
\usepackage{algorithm,algpseudocode}
\usepackage{soul}

\newtheorem{theorem}{Theorem}
\newtheorem{lemma}{Lemma}

\newtheorem{proposition}{Proposition}

\newtheorem{assumption}{Assumption}

\newtheorem{remark}{Remark}
\newenvironment{proof} {\noindent{\it Proof}: }{\hfill$\square$}

\newcommand{\R}{\mathbb{R}}

\begin{document}
	
	\begin{frontmatter}
		
		\title{Sparse Sensing, Communication, and Actuation   via  Self-Triggered  Control Algorithms} 
		
		
		\thanks[infor]{The corresponding author is M.~Bahavarnia.}
		
		\author[UMD]{MirSaleh Bahavarnia \thanksref{infor}}\ead{mbahavar@umd.edu},    
		\author[Lehigh]{Hossein K. Mousavi}\ead{mousavi@lehigh.edu},
		\author[Lehigh]{Nader Motee}\ead{motee@lehigh.edu}

		\address[UMD]{Department of Electrical and Computer Engineering and Institute for Systems Research, University of Maryland at College Park, USA}
		\address[Lehigh]{Department of Mechanical Engineering and Mechanics, Lehigh University, Bethlehem, USA}  

		
\begin{abstract} 
We propose a self-triggered control algorithm to reduce onboard processor usage, communication bandwidth, and energy consumption across a linear time-invariant networked control system.  We formulate an optimal control problem by penalizing the $\ell_0$--measures of the feedback gain and the vector of control inputs and maximizing the dwell time between the consecutive triggering times. It is shown that the corresponding $\ell_1$--relaxation of the optimal control problem is feasible and results in a stabilizing feedback control law with   guaranteed performance bounds, while providing a sparse schedule for collecting samples from sensors, communication with other subsystems, and activating the input actuators.  
		\end{abstract}

	\end{frontmatter}
	
	\section{Introduction}
	
	
	
	The design of feedback control strategies with sparse communication topologies has been one of the active research areas in networked control systems in recent years \cite{arastoo2016closed,fazelnia2017convex,lin2013design,motee2017sparsity}, where the control objective is to find a sparse state feedback gain by minimizing a $\ell_1$-regularized, which is usually quadratic, cost functional. Promoting sparsity in the state feedback gain results in  a reduced communication requirement among the subsystems, which is practically plausible in networked systems with limited onboard resources (e.g., short battery life on quadcopters). In this paper, we take further steps and enhance sparsity by using a self-triggered state feedback control mechanism and allowing subsystems to sense less frequently,  communicate whenever performance deteriorates below some pre-specified threshold, and activate their actuators if necessary.   
	
	The fundamental idea of the self-triggered control is to reduce the need for continuous  sensing and actuation by incurring some performance loss in a controlled manner. A performance-preserving condition or Lyapunov-based stability condition, which is known as  the self-triggering condition, is usually verified to determine whether new updates are  necessary or not. In this approach, the next update  time is solely computed based on the current state information  \cite{heemels2012introduction,gommans2014self}. In  \cite{souza2014self}, the problem of linear time-invariant networked systems with limited bandwidth on their communication channels is considered, where the authors develop a self-triggered method to schedule sampling instances as well as switching times of feedback gains with $\mathcal{H}_2$ and $\mathcal{H}_\infty$ performance guarantees. In \cite{nowzari2012self}, the authors employ ideas from  the self-triggered control to develop an algorithm for a network of multiple agents on how to collect new samples and compute localized control inputs in order to achieve optimal static deployment in a given convex region.   The authors of  \cite{nagahara2016maximum} show that if the sparsest control signal for a finite-horizon optimal control is unique, then the corresponding  $\ell_1$-regularized optimal control problem recovers the sparsest solution. They leverage this idea and design a sparse self-triggered control law for linear time-invariant  systems. In \cite{ito2014optimal}, the $L^p$ control problems with $0 < p <1$  is considered, where the authors obtain a maximum principle and conditions for the existence of solution for optimal control problems with $L^0$-regularized quadratic costs. In \cite{gommans2014self}, the triggering times and feedback gains for the quadratic optimal control  of linear time-invariant systems are designed with gauranteed  performance bounds. 
	
\vspace{-0.2cm}	
	In this paper, we propose a control algorithm that is based on ideas from sparse feedback control and self-triggered control design techniques to achieve higher levels of sparsity, both in time and space, by reducing sensing, communication, and input actuation requirements with a prescribed performance bound. To achieve these goals, in Sections \ref{ProFor} and \ref{sec:formulation}, we formulate an optimal control problem in which subsystems measure their own state variables and then share them over a (time-varying) communication topology. Each subsystem receives  the state information of other  neighboring subsystems once (at triggering times) and forms its optimal control input using these samples and apply it for some period of time without updating it. As soon as the performance deteriorates below some threshold, subsystems trigger and collect new samples and repeat the process. The control inputs are formed using state feedback gains that are computed by minimizing a     
	cost function that is penalized by the $\ell_0$-measures of the feedback gain and  the vector of control inputs. In Section \ref{EqRe}, we show that this optimal control problem and its $\ell_1$-relaxation are feasible. To prolong the triggering intervals, where over which new samples will not be collected and feedback gains stay fixed, our algorithm in Section \ref{ssocsection} solves      
	a nonlinear optimization problem for the next triggering time. Then, using this value, a $\ell_1$-relaxation of the optimal control problem is solved to find a sparse feedback gain. In the last step, the value of the next triggering time is corrected using this feedback gain. It is shown that the resulting feedback control strategy is stabilizing and provides guaranteed performance bounds. The usefulness  of our theoretical findings is verified through extensive simulations in Section \ref{NuSi}.  This paper is an outgrowth of \cite{Bahavarnia17} with several updated and new results and simulations.

	\noindent{\it Notations:} The set of real numbers, positive real numbers, positive integer numbers, and non-negative integer numbers are denoted by $\mathbb{R}$, $\mathbb{R}_{++}$, $\mathbb{N}$, and $\mathbb{Z}_{+}$, respectively.   The supremum of a subset of $\R$ is denoted by $\sup$. For a function of time $f(t)$,  $f'(t)=\partial f(t)/\partial t$. 
	The positive semi-definiteness and positive-definiteness are denoted by $\succeq 0$ and $\succ 0$, respectively. The identity matrix is represented by $I$. The Euclidean norm of vector $v$ is denoted by $\|v\|_2$. The cardinality ($\ell_0$ sparsity measure), $\ell_1$ norm, and largest singular value of a matrix $M$ are represented by $\|M\|_0$, $\|M\|_1$, and $\|M\|$, respectively, where $\|M\|_0$ is identical to its number of nonzero elements and $\|M\|_1$ is the sum of absolute values of its elements. A matrix is called  Hurwitz if all of its eigenvalues have negative real parts. The vectorization and determinant of a matrix $M$ are denoted by $\mathrm{vec}(M)$ and $\det(M)$, respectively. The Kronecker product  is denoted by $\otimes$.
	
	\section{Problem Statement} \label{ProFor}
	We consider  the linear time-invariant (LTI) control systems described by
	\begin{equation}
	\dot{x}(t)= Ax(t)+Bu(t),~~ x(0)=x_0, \label{STLQR}
	\end{equation}
	where $x \in \mathbb{R}^n$ denotes the state vector, $u \in \mathbb{R}^m$ represents the control input,   and $x_0$ is the initial state. To stabilize  system \eqref{STLQR} and regularize $x(t)$ toward the origin, we use the sample-and-hold feedback control laws 
	\begin{equation} \label{KTLQR}
	u(t)
	= u_k:= F_k\, x(t_k)= F_k\, x_k,
	\end{equation}
	for all $t \in [t_k,t_{k+1})$ , in which the sequence of time instants $\{t_k\}_{k=0}^\infty \subset \R_+$ are called triggering times, and  $[t_k,t_{k+1})$ is the $k$'th time interval. The resulting sequence of inter-execution times is $\{\delta_k\}_{k=0}^\infty$, wherein
	\begin{equation}
	\delta_k:=t_{k+1}-t_k.
	\end{equation}
	The sequence of feedback gains is also denoted by $ \{F_k\}_{k=1}^\infty \subset \R^{ m \times n}$. 
	The cost-to-go $J_k(F_k,\xi;x_k)$ corresponding to  time interval $[t_k,t_k+\xi)$ is 
	\begin{align}\label{Jk}
	\begin{adjustbox}{max width=218 pt}
	$
	J_k(F_k,\xi;x_k):= \displaystyle \int_{t_k}^{t_{k}+\xi}\left ( {x(t)}^T Q x(t) + {u(t)}^T R u(t) \right) ~ dt
	$
	\end{adjustbox} 
	\end{align}
	for all $\xi \in [0,\delta_k)$ and weight matrices $Q \succ 0 $ and $R \succ 0$. 
	\begin{assumption} \label{contro-observ}
		The pair $(A,B)$ is stabilizable. 
	\end{assumption}
	
	Since $Q \succ 0$ and Assumption \ref{contro-observ} holds, the LQR problem corresponding to cost functional \eqref{Jk}   has a unique stabilizing solution. 
	For a pre-designed and stabilizing feedback gain $\tilde F$,  the infinite-horizon cost value is given by 
	\begin{align}\label{eq:tildeJ}
	\tilde J(x_0):=x_0^T\, \tilde P\, x_0,
	\end{align}
	where $\tilde P \succ 0$ is the unique  solution to  Lyapunov equation
	\begin{equation} \label{FtildeP}
	\left (A+B\tilde{F}\right )^T \tilde{P} + \tilde{P} \left (A+B\tilde{F}\right ) + Q + \tilde{F}^T R \tilde{F}=0.
	\end{equation}
	This solution exists because $A+B \tilde F$ is  Hurwitz  and $Q + \tilde{F}^T R \tilde{F} \succ 0$. 
	A popular choice for  $\tilde{F}$ is the standard linear quadratic regulator (LQR). In the traditional controller design methods, $\tilde F$ is usually non-sparse; i.e., most of its elements are nonzero. 

	The {\it problem} is to design the sequences of inter-execution times $\left \{\delta_k\right \}_{k=0}^\infty$ and  feedback gains $\left \{F_k\right \}_{k=0}^\infty$  in order to ensure the following control objectives: {\it(i)} sparse sensing both in space and time, 
	{\it(ii)} sparse scheduling of actuators, and  {\it(iii)} guaranteed stability and closed-loop performance compared to a well-performing feedback gain $\tilde F$. 
	
	\vspace{-0.2cm}
	
	\section{Initial Problem Developments}\label{sec:formulation}
	
	We show how the objectives of the paper can be translated to an optimization problem. To achieve  objective {\it(i)}, one can maximize the inter-execution time $\delta_k$ and minimize $ \| F_k \|_{0}$. To realize objective {\it(ii)}, we may minimize $ \| u_k \|_{0}$, which is the number of actuators that is used in interval $[t_k,t_{k+1})$. 
	To achieve objective {\it(iii)}, we consider a Lyapunov-based performance and stability analysis, which is a well-investigated methodology in model predictive control (MPC) (e.g., see \cite{gommans2014self}).  First,  the total cost value is given by \vspace{-1mm} 
	\begin{align}
	& J(x_0):=\sum_{k=0}^{\infty} J_k(F_k,\delta_k;x_k).
	\end{align}
	The relative performance loss  with respect to the benchmark cost value \eqref{eq:tildeJ} is
	\begin{align} \label{eq:nu}
	\nu(x_0):= \frac{J(x_0)-\tilde{J}(x_0)}{\tilde{J}(x_0)} .
	\end{align}

	\begin{theorem} \label{PSTABP}
		Suppose that for some given $\alpha>1$, the sequences of triggering times,  $\{t_k\}_{k=0}^\infty$,  inter-execution times $\left \{\delta_k\right \}_{k=0}^\infty$, and  feedback gains $\left \{F_k\right \}_{k=0}^\infty$, inequality  
		\begin{align}
		J_k(F_k,\xi;x_k) \le \alpha \left(V\big(x(t_k)\big)-V\big(x(t_k+\xi)\big)\right), \label{PerfoC}
		\end{align}
		holds for every $\xi \in [0,\delta_k]$, in which
		the Lyapunov function $V:~\R^n \rightarrow \R$   is defined as 
		$$V(x(t)):=x(t)^T\, \tilde{P}\, x(t).$$
		Then,  control law \eqref{KTLQR} is stabilizing. Moreover, 	for all $x_0 \in \R^n$,  the relative performance loss $\nu(x_0)$ satisfies 
		\begin{align}\label{eq:labeldd}
		\nu(x_0) \, \leq \, \alpha-1.
		\end{align} 
	\end{theorem}
	
	\vspace{-.6cm}
	
	\begin{proof}
		Let us define  
		$
		S_N:= \displaystyle \sum_{k=0}^{N} J_k(F_k,\delta_k;x_k).
		$
		We can use \eqref{PerfoC} and the limit of $\xi \rightarrow \delta_k$ for each term to get
		\begin{align*}
		S_N\,&\leq \, \sum_{k=0}^{N} \alpha \left(V(x_k)-V(x_{k+1})\right)\\
		& =\alpha (V(x_0)-V(x_N)) <\alpha V(x_0)=\alpha \tilde J(x_0). 
		\end{align*}
		Since $\{S_N\}_{N=0}^\infty$ is a nondecreasing sequence that is bounded from above, it has a limit, that is actually $J(x_0)$. Moreover, 
		\begin{align}\label{StIn}
		J(x_0)=\lim_{N \rightarrow \infty } S_N \leq \alpha \tilde J(x_0). 
		\end{align}
		A simple reorganization of (\ref{StIn}) reveals that \eqref{eq:labeldd} holds. 
	\end{proof}



	We combine these objectives and formulate the self-triggered sparse optimal control   problem as  
	\begin{flalign}
	&\underset{F_k,~\delta_k}{\mathrm{minimize}}~~~~~~ -\delta_k+ \gamma \| F_k \|_{0} + \eta \| u_k \|_{0} \tag{\bf{P1}} \label{EQP1} && 
	\\
	& \mathrm{subject~to}:~~~ (\ref{STLQR})~~\mathrm{and}~(\ref{KTLQR}), \notag
	&& \\
	& ~~~~~~~~~~~~~~~~~~~\mathrm{~\eqref{PerfoC}} \mathrm{~for~all~}  \xi \in [0,\delta_k]  && \notag 
	\end{flalign}
	The nonnegative parameters $\gamma$ and $\eta$ adjust the balance between the levels of spatial and temporal sparsity in terms of sensing and actuation, respectively. 
	
	\begin{theorem} 
		Every feasible solution of problem \eqref{EQP1}  is  a stabilizing control law with performance guarantee  \eqref{eq:labeldd}.  
	\end{theorem}

	
	\begin{remark}
		The method proposed by \cite{heemels2012introduction} solely takes advantage of the Lyapunov function evolutions. However, our proposed problem setup resembles the method proposed by \cite{gommans2014self}, in which in addition to the Lyapunov functions, a performance-based condition is utilized. 
	\end{remark}
	

	\section{Equivalent Formulation, Relaxation, and Feasibility } \label{EqRe}
	We show that \eqref{EQP1} can be reformulated as a regularized quadratically-constrained quadratic program (QCQP) when $\delta_k$ is kept fixed. Then, we suggest  appropriate relaxations that make the problem more   tractable. 
	
	
	\begin{lemma}For all $k \in \mathbb{N}$ and $j \in \mathbb{Z}_{+}$,  feedback control law (\ref{KTLQR}) can be decomposed as
		\begin{equation*}
		u(t)=F_k N_k x_0
		\end{equation*}
		for all $t \in [t_k,t_{k+1})$, where $N_0:= I$ and $N_k$ is recursively given by
		\begin{align*}
		N_k & :=  M_{k-1}({\delta}_{k-1})  \hspace{0.0cm} N_{k-1}
		\end{align*}
		and the matrices $M_j$'s are given by 
		$$
		M_j(\xi) := \e^{A \xi} \big(I+Z(\xi)BF_j\big), ~Z(\xi): =  \int_{0}^{\xi} \e^{-A \tau} d\tau.
		$$
	\end{lemma} 
	\begin{proof} Solving system (\ref{STLQR}) together with (\ref{KTLQR}) for  time interval $[t_k,t_{k+1})$ gives us \vspace{-5mm}
		\begin{align*}
		x(t) &= \e^{A(t-t_k)}x(t_k) + \int_{t_k}^{t} \e^{A(t-\phi)}BF_k  x(t_k) d\phi\\
		&= \e^{A(t-t_k)} \left (I + \int_{0}^{t-t_k} \e^{A(-\tau)}BF_k~d\tau \right ) x(t_k)\\
		& =  \e^{A(t-t_k)} \big(I + Z(t-t_k) BF_k \big) x(t_k)\\
		& =  M_k(t-t_k) x(t_k).
		\end{align*}
		Since $x(t)$ is continuous  at 
		$t=t_{k+1}$, thus  $$x(t_{k+1})=\e^{A \delta_k} \big(I + Z(\delta_k) BF_k \big) x(t_k)=M_k(\delta_k)x(t_k).$$ By induction,  $x(t_k)=N_k x_0$.
	\end{proof}

	We derive a quadratic form for the cost-to-go. 
	
	\begin{lemma} \label{propJ}
		The $k$'{th} time interval cost $J_k(F_k,\xi;x_k)$ can be expressed as the quadratic term 
		\begin{equation*}
		J_k(F_k,\xi;x_k)=x_k^T Y_k(F_k,\xi) x_k,
		\end{equation*}
		where the matrix $Y_k$ is given by
		\begin{align*} \begin{adjustbox}{max width=245 pt}
		$
		Y_k(F_k,\xi) := H_0(\xi)+F_k^T H_1(\xi)^T + H_1(\xi) F_k +  F_k^T H_2(\xi) F_k,
		$
		\end{adjustbox}
		\end{align*}
		for all $k \in \mathbb{Z}_{+}$ and 
		\begin{align*}
		H_0(\xi) & :=  \int_{0}^{\xi} \e^{A^T \tau}Q\e^{A \tau}d\tau, \\
		H_1(\xi) & :=  \int_{0}^{\xi} \e^{A^T \tau}Q \e^{A \tau}Z(\tau)Bd\tau, \\
		H_2(\xi) & :=  \int_{0}^{\xi} \big (\e^{A \tau}Z(\tau)B \big)^T Q \big(\e^{A \tau}Z(\tau)B \big)d\tau + \xi R.
		\end{align*}
	\end{lemma}
	\vspace{-0.5cm}
	\begin{proof}
		We can write 
		\begin{equation*}
		J_k(F_k,\xi;x_k)=J_k^{x}(F_k,\xi;x_k)+J_k^{u}(F_k,\xi;x_k),
		\end{equation*}
		wherein
		\begin{align*}
		J_k^{x}(F_k,\xi;x_k) &= \int_{t_k}^{t_{k}+\xi} x_k^T M_k(t-t_k)^T Q M_k(t-t_k) x_k dt  \\
		& =  x_k^T \int_{0}^{\xi} M_k(\tau)^T Q M_k(\tau)d \tau x_k 
		\end{align*}
		and \vspace{-4mm}
		\begin{align*} \begin{adjustbox}{max width=240 pt}
		$
		J_k^{u}(F_k,\xi;x_k) = \displaystyle \int_{t_k}^{t_{k}+\xi} x_k^T F_k^T R F_k x_k dt =  x_k^T (\xi F_k^T R F_k) x_k .
		$
		\end{adjustbox}
		\end{align*}Then, we get $J_k(F_k,\xi;x_k)= x_k^T Y_k(F_k,\xi) x_k$, wherein
		\begin{equation} \label{Yintegral}
		Y_k(F_k,\xi)= \int_{0}^{\xi} M_k(\tau)^T Q M_k(\tau)d\tau + \xi_k F_k^T R F_k.
		\end{equation}
The claim follows by substituting $M_k(\tau)=\e^{A \tau} \big(I+Z(\tau)BF_k \big)$ in the right side of \eqref{Yintegral}. 
	\end{proof}
	
	When $\delta_k$ is kept fixed, the following proposition states an important property about matrix $H_2(\xi)$ which leads to the convex reformulation of  \eqref{EQP1} (in terms of  variable $F_k$) .
	\vspace{0.1cm}
	\begin{proposition} \label{H2del}
		For all $\xi  >0$,  $H_2(\xi)$ is positive-definite. 
	\end{proposition}
	
	\begin{proof}
		For arbitrary nonzero vector $v \in \mathbb{R}^m$, we get 
		\begin{align*}
		&v^TH_2(\xi)v \\ 
		&= v^T \left( \int_{0}^{\xi} \big (e^{A \tau}Z(\tau)B \big)^T Q \big(e^{A \tau}Z(\tau)B \big)d\tau + \xi R \right ) v\\
		&= \int_{0}^{\xi} \big (e^{A \tau}Z(\tau)B v \big)^T Q \big(e^{A \tau}Z(\tau)B v \big)d\tau + \xi v^T R v.
		\end{align*}Since $Q \succ 0$, the last integrand is non-negative. Therefore, the integral is also non-negative. Since $R \succ 0$, $\xi v^T R v>0$. Thus, $v^T H_2(\xi)v > 0$ and the conclusion follows. 
	\end{proof}
	
	Now, we can state the following equivalent formulation. 
	
	\begin{theorem} \label{propo2}
		The  optimization problem \eqref{EQP1} is equivalent to 
		\begin{flalign}
		&\underset{f_k,\delta_k}{\mathrm{minimize}}~~
		-\delta_k+ \gamma \| f_k \|_{0} + \eta \| (x_k^T \otimes I) f_k \|_{0} \tag{\bf{P2}} && \label{quadobject}   
		\\
		& \mathrm{subject~to}:  \mathrm{~for~all~} \xi \in [0,\delta_k): \label{quadcon} \\
		& ~~~~~~~~~~~~~~~~~~\frac{1}{2}f_k^T P_1(\xi) f_k + q_1(\xi)^T f_k + r_1(\xi) \le 0, \notag 
		\end{flalign}
		where $f_k:=\mathrm{vec}(F_k)$ and 
		\begin{align*}
		&\begin{adjustbox}{max width=245 pt}$
		P_1(\xi) :=2(x_k x_k^T) \otimes \big(H_2(\xi)+\alpha B^T Z(\xi)^T e^{A^T \xi}\tilde{P}e^{A \xi} Z(\xi) B \big),$
		\end{adjustbox} \\
		& \begin{adjustbox}{max width=245 pt}$
		q_1(\xi) := 2 \mathrm{vec}\bigg(\big(H_1(\xi)^T+\alpha B^TZ(\xi)^T e^{A^T \xi}\tilde{P}e^{A \xi}\big)x_k x_k^T \bigg),$
		\end{adjustbox}\\
		& \begin{adjustbox}{max width=245 pt}$
		r_1(\xi) := x_k^T \big(H_0(\xi)+ \alpha(e^{A^T \xi}\tilde{P}e^{A \xi}-\tilde{P})\big)x_k.
		$
		\end{adjustbox}
		\end{align*}
	\end{theorem}
	\vspace{-0.5cm}
	\begin{proof}
		The building block of the proof is the identity
		\begin{equation} \label{eq:id}
		\mathrm{vec}(UVW)= (W^T \otimes U )\mathrm{vec}(V)
		\end{equation}
		for any triplet $(U,V,W)$. Using \eqref{eq:id} for $U= I$, $V=F_k$, and $W=x_k$, the objective function of \eqref{EQP1} takes the form of objective function in \eqref{quadobject}. Likewise, by repeated application of \eqref{eq:id} on    \eqref{PerfoC}, one can verify \eqref{quadcon}.
	\end{proof}
	
	If we combine Theorem \ref{propo2} and  of Proposition \ref{H2del}, we observe that constraint \eqref{quadcon} involves a quadratic convex function of variable $f_k=\mathrm{vec}(F_k)$. Thus, when $\delta_k$ is kept fixed, \eqref{quadobject} takes the form of a regularized QCQP. Next, we relax \eqref{quadobject} by replacing the $\ell_0$-norm with its convex surrogate  the $\ell_1$-norm to get the following form. 
	\begin{flalign}
	&\underset{f_k,\delta_k}{\mathrm{minimize}}~~
	-\delta_k+ \gamma \| f_k \|_{1} + \eta \| (x_k^T \otimes I) f_k \|_{1}\tag{\bf P3} \label{quadob} &&  
	\\
	& \mathrm{subject~to}:  \mathrm{~for~all~} \xi \in [0,\delta_k]: \label{quadcon} \\
	& ~~~~~~~~~~~~~~~~~~\frac{1}{2}f_k^T P_1(\xi) f_k + q_1(\xi)^T f_k + r_1(\xi) \le 0, \notag 
	\end{flalign}
	Although the $\ell_1$-relaxation makes \eqref{quadobject} more tractable,  the left hand side of constraint \eqref{quadob} may be still a nonconvex function of $\xi$. Thus, solving \eqref{quadob} for optimal solutions generally remains a difficult non-convex task. In order to sub-optimally solve \eqref{quadob} we propose  a two-stage method that is discussed in Section \ref{ssocsection}. The constraints of relaxed problem \eqref{quadob} and the original problem \eqref{EQP1} are identical. Moreover, stability and performance guarantees of  \eqref{EQP1} originate from its constraint. Hence,  their feasibility, stability, and performance guarantee are     equivalent. 
	
	\begin{theorem} \label{thm:eq}
		The feasible solutions to  \eqref{quadob} are feasible solutions to \eqref{EQP1}. Moreover, any solution to   \eqref{quadob}  results in a stabilizing control \eqref{KTLQR} law with performance guarantee  \eqref{eq:labeldd}.  
	\end{theorem}

	We prove the feasibility of these optimization problems.
	\begin{theorem}\label{Feasi} 
		For every $\alpha>1$, problem \eqref{EQP1} and its relaxation \eqref{quadob} are feasible.
	\end{theorem}
	\begin{proof}  \eqref{PerfoC} is equivalent to $g(\xi;F_k)\leq 0$, where 
\begin{equation*} 
\resizebox{1.02\hsize}{!}{$g(\xi;F_k):=  
		x_k^T \left(Y_k(F_k,\xi)+ \alpha \big(-\tilde{P}+M_k(\xi)^T\tilde{P}M_k(\xi)\big)\right) x_k.$}
\end{equation*}
		We can show that $g'(0;F_k)=x_k^T W_k x_k$, in which
		\begin{equation*}
		W_k=Q+F_k^TRF_k + \alpha \big((A+BF_k)^T\tilde{P}+\tilde{P}(A+BF_k)\big).
		\end{equation*}According to (\ref{FtildeP}), we get
		\begin{equation*}
		g'(0;\tilde{F})=(1-\alpha)x_k^T \left (Q+\tilde{F}^TR\tilde{F}\right )x_k. 
		\end{equation*}
		Since $\alpha > 1$ and $Q+\tilde{F}^TR\tilde{F}\succ 0$ , we get $g'(0;\tilde{F}) < 0$. In addition, $g(0;\tilde F)=0$.  Therefore, there exists a positive $\theta_k$ at $F_k=\tilde{F}$ such that  for all $\xi \in (0,\theta_k)$ we have  $g(\xi;\tilde{F}) < 0$. This proves the feasibility of \eqref{EQP1}. Using Theorem \ref{thm:eq} the feasibility of   \eqref{quadob} is followed as well.
	\end{proof}
	
	
	%
	%

	\section{Self-Triggered Sparse Optimal Control} \label{ssocsection}
	
	We  find a sequence of stabilizing feedback gains and   inter-execution times that satisfy  desired performance guarantees. 
	
	\subsection{Solving \eqref{quadob} for $\delta_k$ while $F_k$ is   Fixed} \label{FirStep}
	This concerns a subproblem that is  a component of our   algorithm. Inequality (\ref{quadcon}) is equivalent to the  constraint
	\begin{equation} \label{equpper}
	\frac{1}{2} u_k^T P_2(\xi) u_k + q_2(\xi)^T u_k + r_1(\xi) \le 0, 
	\end{equation}
	where $u_k: = F_k x_k$ and
	\begin{align*}
	P_2(\xi) &:= 2H_2(\xi)+2\alpha B^T Z(\xi)^T \e^{A^T \xi}\tilde{P}e^{A \xi} Z(\xi) B,\\
	q_2(\xi) &:=  \left (2H_1(\xi)^T+2\alpha B^TZ(\xi)^T \e^{A^T \xi}\tilde{P}e^{A \xi} \right ) x_k.
	\end{align*}
	When we fix feedback gain $F_k$, \eqref{quadob} becomes equivalent to the nonlinear optimization problem
	\begin{flalign}
	&  \underset{\delta_k}{\mathrm{maximize}}~ ~\delta_k \tag{\bf P4} \label{EQP4} && \\
	& \mathrm{subject~to:} ~\eqref{equpper}   \mathrm{~for~every~}  \xi \in [0,\delta_k]   . &&\notag 
	\end{flalign}
	The optimal solution of \eqref{EQP4}  can be characterized by  
	\begin{align} \label{omegamax}
	\delta_k^* =  \sup \left. \Big \{\theta_k>0~ \right | (\ref{equpper})~ \mathrm{ holds~for~all~}  \xi \in [0,\theta_k) \Big \}.
	\end{align}
	In order to solve (\ref{omegamax}), we utilize the simple and commonly-used discretization rule with fixed step-size, e.g., see \cite{mazo2009self,heemels2012introduction}.  The general nonlinear optimization  methods can be effective to solve (\ref{omegamax}) as well. 
	In the rest of this paper, we suppose that there exists a procedure $\mathrm{InterExec}(F_k)$  that solves \eqref{EQP4}  for a fixed gain $F_k$ and returns $\delta_k^*$. 
	
	\subsection{Solving for $F_k$  when $\delta_k$ is  Fixed}
	
	When $\delta_k$ is kept fixed, solving \eqref{quadob} for $F_k$ is still an infinite-dimensional optimization. Therefore, we relax the problem 
	by evaluating \eqref{equpper} only at $\xi=\delta_k$ and ensuring the feasibility of 
	\eqref{EQP4} after updating $F_k$ by imposing an additional constraint.
	Applying the Schur complement to \eqref{equpper} and setting the argument $\xi$ to $\delta_k$, we arrive at the LMI constraint 
	\begin{equation} \label{eqlowerform}
	\begin{bmatrix}
	2P_2(\delta_k)^{-1} & F_k x_k \\x_k^T F_k^T & -q_2(\delta_k)^T F_k x_k- r_1(\delta_k)
	\end{bmatrix} \succeq 0.
	\end{equation}
	
	%
	%
	
	\begin{proposition} \label{necsuf} Suppose that $F_k$ satisfies 		\begin{equation} \label{SchurCondition_0}
		\begin{adjustbox}{max width=213 pt}$
		\begin{bmatrix}
		R^{-1} & F_k \\ F_k^T & -\alpha \big((A+BF_k)^T\tilde{P}+\tilde{P}(A+BF_k)\big)-Q
		\end{bmatrix} \succ 0.$
		\end{adjustbox}    
		\end{equation}
		Then, \eqref{EQP4} for feedback gain $F_k$ is feasible.
	\end{proposition}
	\begin{proof}
		If $g'(0;F_k) < 0$, then using the same argument given in the proof of Theorem \ref{Feasi}, \eqref{EQP4} would be feasible.  
		One inspects that inequality $g'(0;F_k)  < 0$  is equivalent  to  
		\begin{equation*}\begin{adjustbox}{max width=240 pt}$
		x_k^T \big(Q+F_k^TRF_k + \alpha \big((A+BF_k)^T\tilde{P}+\tilde{P}(A+BF_k)\big) \big)x_k < 0.
		$ 
		\end{adjustbox}
		\end{equation*}
		For this inequality to hold independent of $x_k$, we need 
		\begin{equation}
		Q+F_k^TRF_k + \alpha \big((A+BF_k)^T\tilde{P}+\tilde{P}(A+BF_k)\big) \prec 0. \label{MInec}
		\end{equation}
		Applying the Schur complement to (\ref{MInec}), we obtain (\ref{SchurCondition_0}).
	\end{proof}
	
	
	Solving the two LMIs \eqref{eqlowerform} and LMI~\eqref{SchurCondition_0} together with our objective function, we arrive at a   semi-definite program 
	\begin{flalign}
	& \hspace{-0cm}F_k^*= \hspace{0cm} \underset{F_k}{\mathrm{argmin.}}~~ \gamma \| F_k \|_{1} + \eta \|F_k x_k\|_{1} && \tag{\bf P5} \label{EQP6} \\
	& \hspace{0.95cm} \mathrm{subject~to: ~ LMI~\eqref{eqlowerform}~ and ~LMI~\eqref{SchurCondition_0},}  && \notag
	& \hspace{-0cm}
	\end{flalign}
	The gain $F_k$ that is obtained from \eqref{EQP6} may not satisfy inequality (\ref{quadcon}) for all $\xi \in [0,\delta_k]$ that may arise  from ignoring the continuous satisfaction of the  constraint. To remedy this, we leverage  \eqref{omegamax} and modify the value of   $\delta_k$ for the updated value of $F_k$  (see the next subsection). In the rest of this paper, we suppose that there exists a procedure $\mathrm{FeedbackGain}(\delta_k)$  that solves \eqref{EQP6}  for a fixed $\delta_k$ and returns $F_k^*$.
	
	\begin{algorithm}[t]
		\caption{Self-Triggered Sparse Optimal Control}
		\label{alg:ssoc}
		\begin{algorithmic}  
			\algnewcommand\algorithmicinitz{\textbf{initialize:}}
			\algnewcommand\Init{\item[\algorithmicinitz]}
			\renewcommand{\algorithmicrequire}{\textbf{input:}}
			\renewcommand{\algorithmicensure}{\textbf{output:}} 
			\Require $A$, $B$, $Q$, $R$,    $\gamma$, $\eta$, and $\alpha$. \vspace{1mm}
			\Ensure $\{F_k\}_{k=0}^{\infty}$, $\{\delta_k\}_{k=0}^{ \infty }$, and $\{t_k\}_{k=0}^{ \infty }$\vspace{1mm}
			\Init  $t_0=0$, $N_0 \leftarrow I$, evaluate  $\e^{A \xi}$, $H_0(\xi)$, $H_1(\xi)$, $H_2(\xi)$, and $Z(\xi)$ offline  and save them in a look-up table \vspace{0.15cm} 
			\For {$k=0$ to $\infty$} \vspace{1mm}
			\State $F_k \leftarrow \tilde F$  \vspace{1mm}
			\State  $\delta_k \leftarrow \mathrm{InterExec}(F_k)$  \vspace{1mm}  
			\State $F_k \leftarrow \mathrm{FeedbackGain}(\delta_k)$   \vspace{1mm}   
			\State 	$\delta_k \leftarrow \mathrm{InterExec}(F_k)$  \vspace{1mm}    
			\State $M_k(\delta_k) \leftarrow \e^{A \delta_k} (I+Z(\delta_k)BF_k)$ \vspace{1mm}
			\State  $N_{k+1}\leftarrow M_k(\delta_k)N_k$, ~$t_{k+1} \leftarrow t_{k}+\delta_k$  \vspace{1mm}
			\vspace{0.05cm}
			\EndFor
		\end{algorithmic}
	\end{algorithm}
	
	\subsection{Algorithm}
	
	Based on the discussions of the previous subsection, our control design algorithm consists of the following three steps:  {\it (i)}  solve \eqref{EQP4}  to find the first estimate of $\delta_k$ for fixed $F_k=\tilde F$,  {\it (ii)} for a fixed value of $\delta_k$,   find a value for $F_k$ using \eqref{EQP6}, and {\it (iii)}  for the updated value of feedback gain $F_k$,   correct the value of $\delta_k=\delta_k^*$ using \eqref{EQP4}. These steps are summarized in Algorithm \ref{alg:ssoc}. 
	
	\begin{theorem} The sequences of feedback gains $\{F_k\}_{k=0}^\infty$ and triggering times $\{t_k\}_{k=0}^\infty$  resulting from   Algorithm 1 gives a stabilizing control law \begin{equation}  
		u(t)= F_k\, x_k,
		\end{equation}
		for all $t \in [t_k,t_{k+1})$. Moreover, for all $x_0 \in \R^n$,  the relative performance loss $\nu(x_0)$ satisfies 
		\begin{align} \label{eq:bnd} 
		\nu(x_0)\, \leq \, \alpha-1.
		\end{align} 
	\end{theorem}
	\begin{proof}
		Because LMI \eqref{SchurCondition_0} has been incorporated, for the derived value of $F_k$, $g'(0;F_k)<0$, thus, the corresponding value of $\delta_k^*$ is strictly positive. Since for $F_k$ and $\delta_k=\delta_k^*$, the  constraint of the problem \eqref{equpper} is satisfied, \eqref{eq:bnd} holds. 
	\end{proof}
	
	Optimal control problem \eqref{EQP6}  can be solved by a convex optimization toolbox such as CVX \cite{grant2008cvx}. 

	
	\begin{remark} Matrices $\mathrm{e}^{A\xi}$, $Z(\xi)$, $H_0(\xi)$, $H_1(\xi)$ and $H_2(\xi)$ can be computed offline and over a grid with arbitrary resolution, as their values only depend on the state-space matrices, $A$ and $B$, weight matrices, $Q$ and $R$, and  time argument, $\xi$. Then, during the execution, we can use look-up tables to enhance the computational complexity of our algorithm. 
	\end{remark}
	
	\section{Numerical Simulations} \label{NuSi}
	\begin{figure}[t]
		\centering
		\vspace{.05cm}
		{
			\includegraphics[ width=5.3cm]{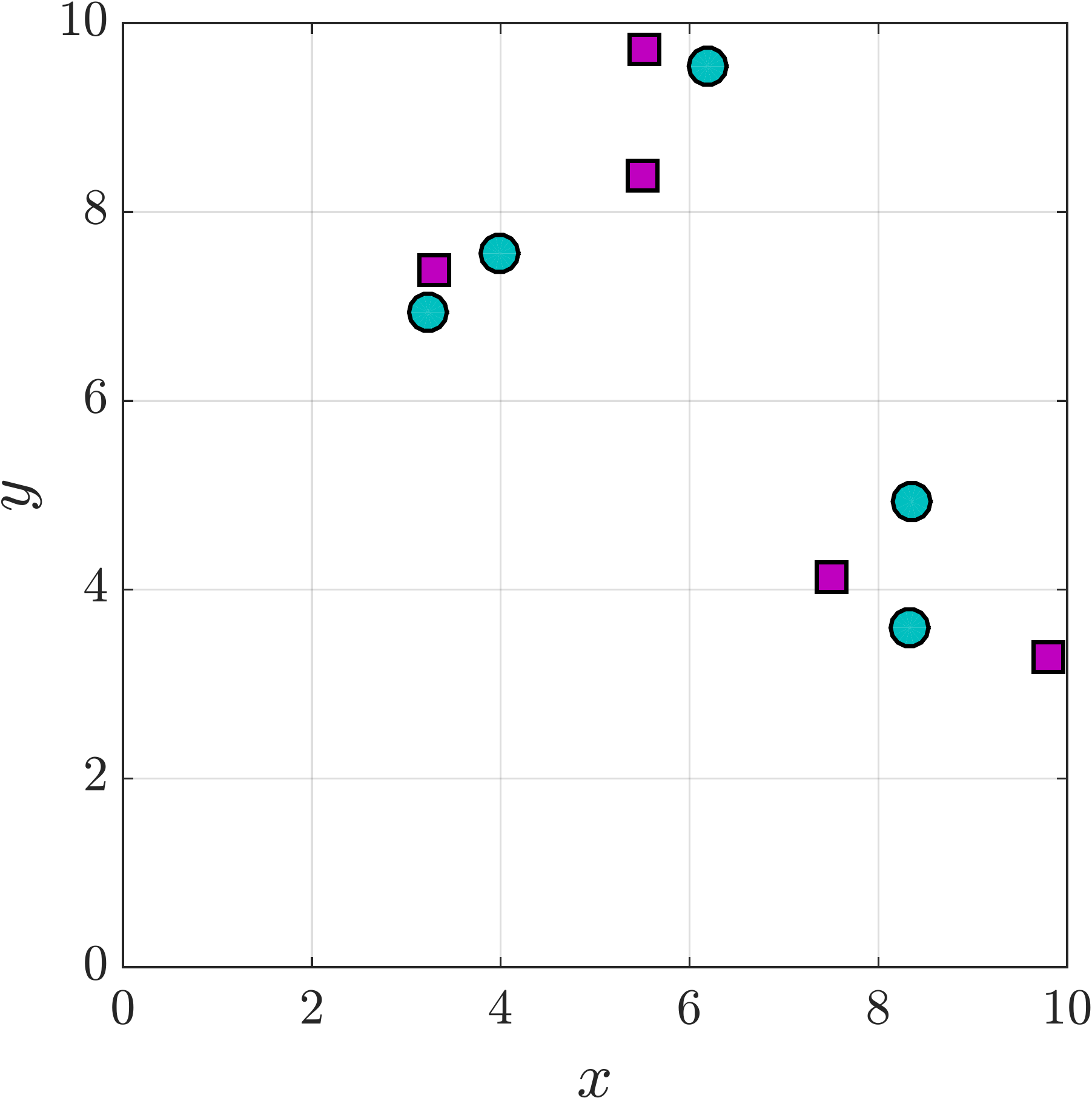}
		}
		\caption{\small Randomly generated positions of $N=10$ subsystems.}
		\label{figpos}
	\end{figure}
	To assess the effectiveness of the proposed  procedure, we consider a class of spatially distributed systems (see \cite{motee2008optimal} for more details).  
	We consider a system of $N=10$ subsystems, which are randomly placed in the region $[0,10] \times [0,10]$; see Fig. \ref{figpos}, in which each shape represents a subsystem. Then, the dynamics of  $i$'{th} subsystem is characterized by \vspace{-3mm}
	\begin{equation*}
	\dot{x}^{(i)}(t) =  [A]_{ii} x^{(i)}(t) + \sum_{j=1,\atop j \neq i}^{N} [A]_{ij} x^{(j)}(t) + [B]_{ii} u^{(i)}(t),\vspace{-5mm}
	\end{equation*}
	where \vspace{-5mm}
	\begin{align*}
	&[A]_{ii} = \begin{bmatrix}  1 & ~~~~\hspace{0.05cm}1 \\ 1 &~~~~\hspace{0.05cm} 2\end{bmatrix}, [B]_{ii} =\begin{bmatrix}  0 \\ 1 \end{bmatrix}~ \textrm{for square shapes},\\
	&[A]_{ii} = \begin{bmatrix}  -2 & 1\\1 & -3\end{bmatrix}, [B]_{ii} =\begin{bmatrix}  0 \\ 1 \end{bmatrix}~ \textrm{for circle shapes}, \\
	&[A]_{ij}= \frac{1}{\e^{\beta \mathrm{dis}(i,j)}} \begin{bmatrix}  1 & ~~~~0 \\ 0 & ~~~~1\end{bmatrix}, [B]_{ij} =\begin{bmatrix}  0 \\ 0 \end{bmatrix},~\forall j\neq i,
	\end{align*}
	where $\beta$ determines the rate of decay in  the couplings and $\mathrm{dis}(i,j)$ denotes the Euclidean distance between nodes $i$ and $j$ in $\R^2$. The resulting system is of dimension $n=20$. We set  $\beta=1$, $\gamma=\eta=0.001$, $\alpha=1.15$, $Q= I$, and $R=2 I$. We apply Algorithm 1 with $k_{\mbox{{\tiny max}}}=49$ iterations in the \textbf{for} loop  and obtain Figures \ref{figrho0} and \ref{figrhoN0}, which depict the relative cardinality of feedback gains given by  
	\begin{equation} \label{eq:kappak}
	\kappa_k:=100 \, {\big \|F_k \big \|_{0}} /{\big \|\tilde F \big \|_{0}},
	\end{equation}
	and relative cardinality of control input vectors: 
	\begin{equation}\label{eq:kappau}
	\mu_k:=100\,  {\big \|u_k \big \|_{0}}/{\big \|\tilde u_k \big \|_0},
	\end{equation}
	respectively. 
	Figures \ref{figrho0} and \ref{figrhoN0} illustrate that 
	this method improves the spatial sparsity compared to the periodic feedback control with $F_k=\tilde F$ for all time intervals. In average, the cardinalities of  $F_k$ and $u_k$ have almost decreased by $38\%$ and $50\%$, respectively, while we lose about $15\%$ in performance. Fig. \ref{figro0} compares  $\|x\|_2$ for two designs: our method and periodic time-triggered LQR design. In Fig. \ref{figroN0}, we illustrate the 
	state  of the system using the time-triggered design, starting from an initial state. Fig. \ref{figvs3} shows the inter-execution times $\delta_k$ versus time $t$. To measure the average   relative cardinality of feedback gains over time, we define   \vspace{-3mm}
	\begin{equation*}
	{R}_F:=  {\displaystyle \sum_{k=0}^{k_{\mbox{{\tiny max}}}}\delta_k \kappa_k } \bigg / {\displaystyle \sum_{k=0}^{k_{\mbox{{\tiny max}}}}\delta_k } ,
	\end{equation*}
	which represents the average sensor activity throughout time. Likewise, the average relative cardinality of control inputs  over time is defined as  \vspace{-3mm}
	\begin{equation*}
	{R}_u:= {\displaystyle \sum_{k=0}^{k_{\mbox{{\tiny max}}}}\delta_k \mu_k } \bigg / {\displaystyle \sum_{k=0}^{k_{\mbox{{\tiny max}}}}\delta_k } ,
	\end{equation*}
	which  represents the average actuator activity throughout time. 
	Also, the average inter-execution time  is   $
	{D}:= \big( \sum_{k=0}^{k_{\mbox{{\tiny max}}}} \delta_k \big) / k_{\max}, $
	 that quantifies the time sparsity of (state) sampling. 
	The dependency of these quantities on $\alpha$ is captured in Table \ref{TDep}, which demonstrates that by accepting higher levels of  performance loss, we achieve sparser control designs in terms of both sensing and actuation activity (i.e., smaller values of $R_F$ and $R_u$) and sampling activity in time (i.e., larger values of $D$). In the next simulations, we preserve the  parameters 
	except for $\beta$. Figures \ref{FigRF}, \ref{FigRU}, and \ref{FigD} demonstrate the dependency of $R_F$, $R_u$, and $D$ on $\beta$, respectively. Fig. \ref{FigRF} and Fig. \ref{FigRU} depict that there is a tradeoff between the spatial decay rate $\beta$ and the indices of sensing and actuating activity, $R_F$ and $R_u$: as $\beta$ increases, the   system tends to be more localized and, consequently,  $F_k$'s and   $u_k$'s  become sparser. Moreover, according to Fig. \ref{FigD}, as the  system gets more localized, the average inter-execution time increases; i.e., less number of samples are  required on average (in time).
	
	\begin{figure}[t]
		\centering
		
		\subfloat[\label{figrho0}]{
			\includegraphics[width=7.1cm]{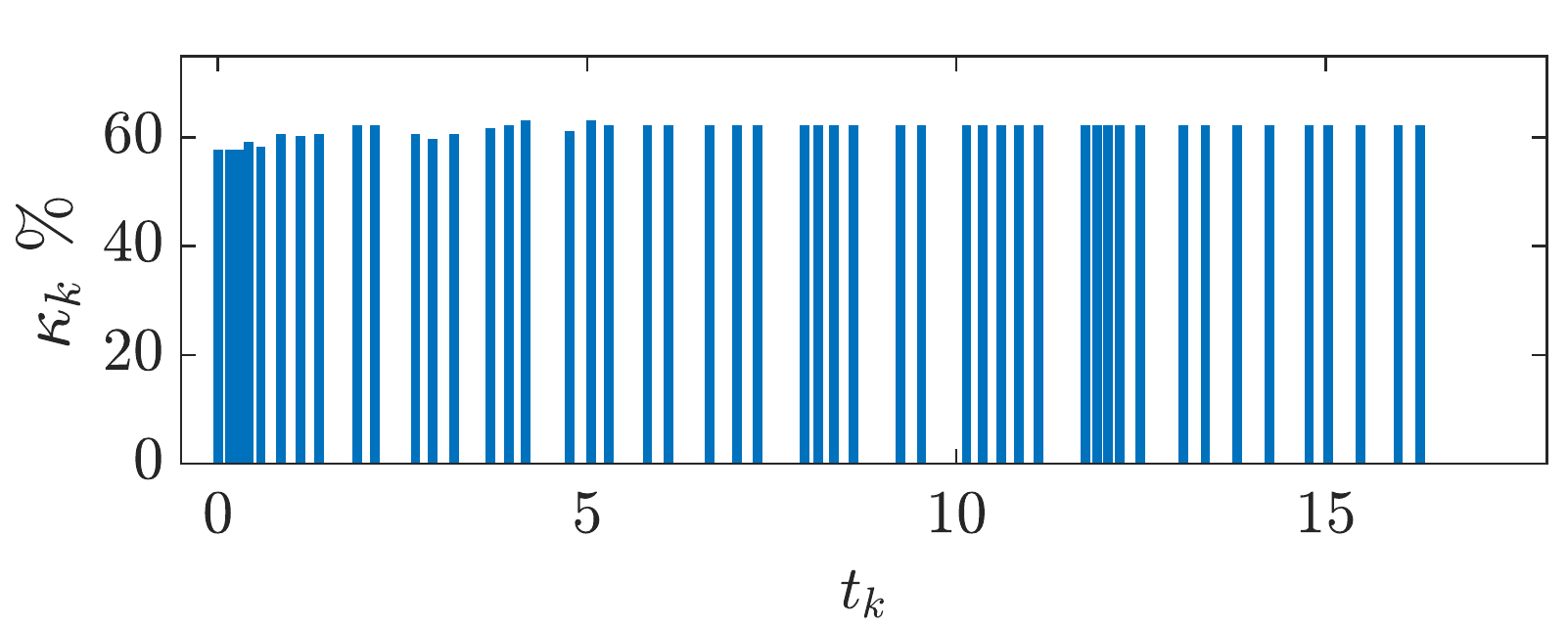}
		}
		\hfill
		\subfloat[\label{figrhoN0}]{
			\includegraphics[width=7.1cm]{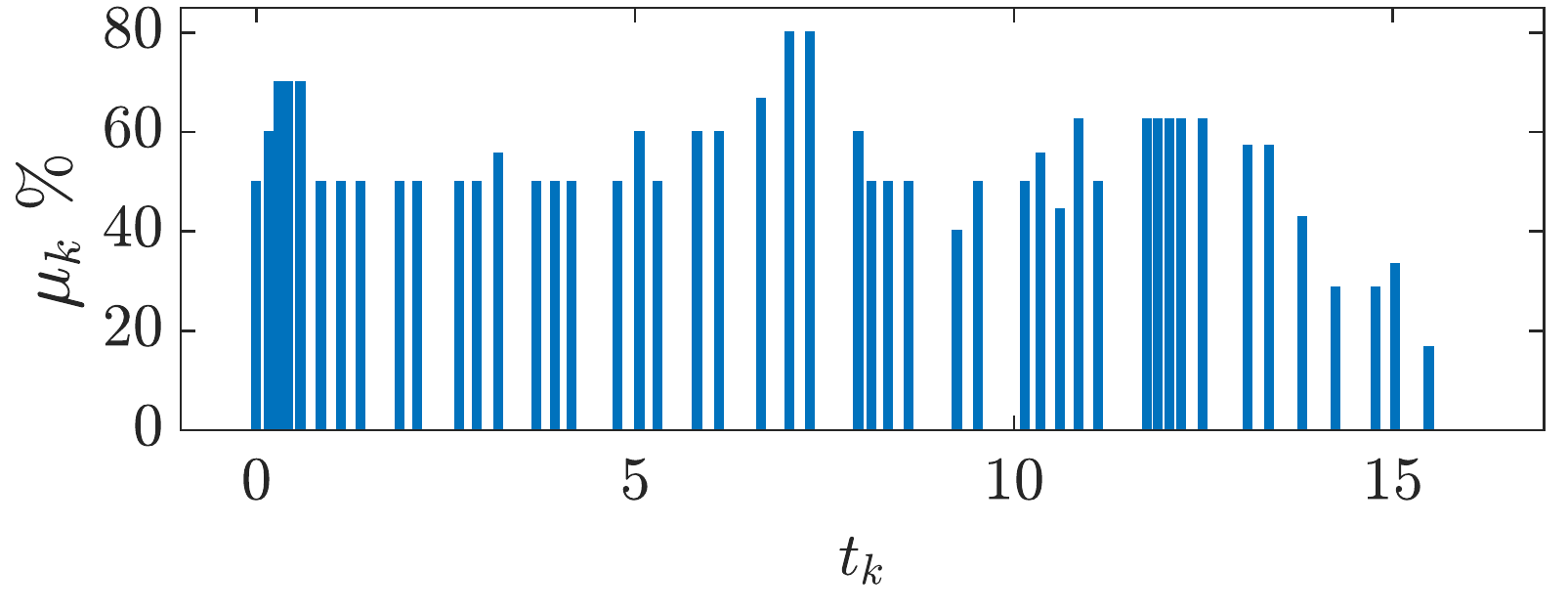}
		}
		\caption{\small (a) Relative cardinality of controllers $\kappa_k$ defined in \eqref{eq:kappak} versus triggering times $t_k$ (b) Relative cardinality of control inputs $\mu_k$ defined in \eqref{eq:kappau} versus triggering times $t_k$.}
		\label{figvs_lambda}
	\end{figure}
	
	\begin{figure}[t]
		\centering
		\subfloat[\label{figro0}]{
			\includegraphics[width=6.7cm]{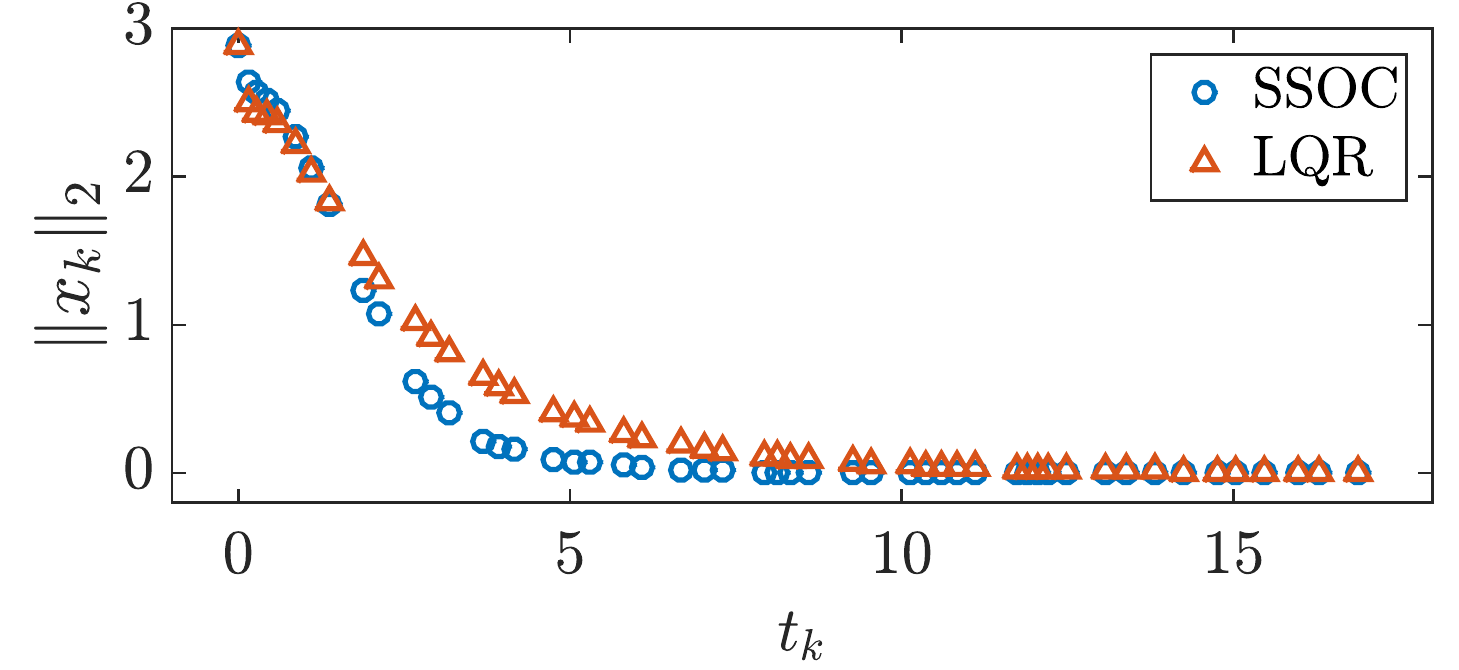}
		}
		\hfill
		\subfloat[\label{figroN0}]{
			\includegraphics[width=6.7cm]{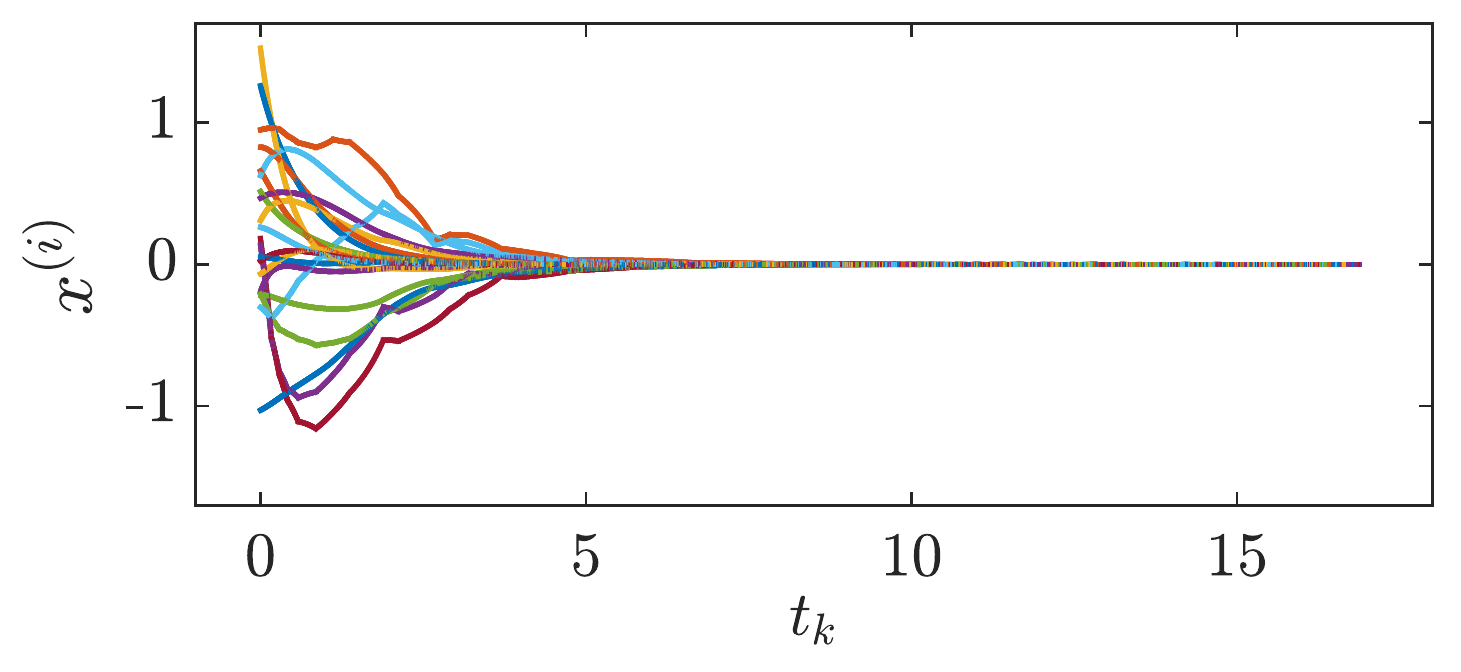}
		}
		\caption{\small (a) The Euclidean norm of state trajectories of SSOC and periodic time-triggered LQR design $\|x_k\|_2$ versus triggering times $t_k$ (b) state trajectories $x^{(i)}$'s of our design with  a random  $x_0$}
		\label{figvs_lambda2}
	\end{figure}
	
	\begin{figure}[t]
		\centering
		\vspace{.2cm}
		{
			\includegraphics[width=6.8cm]{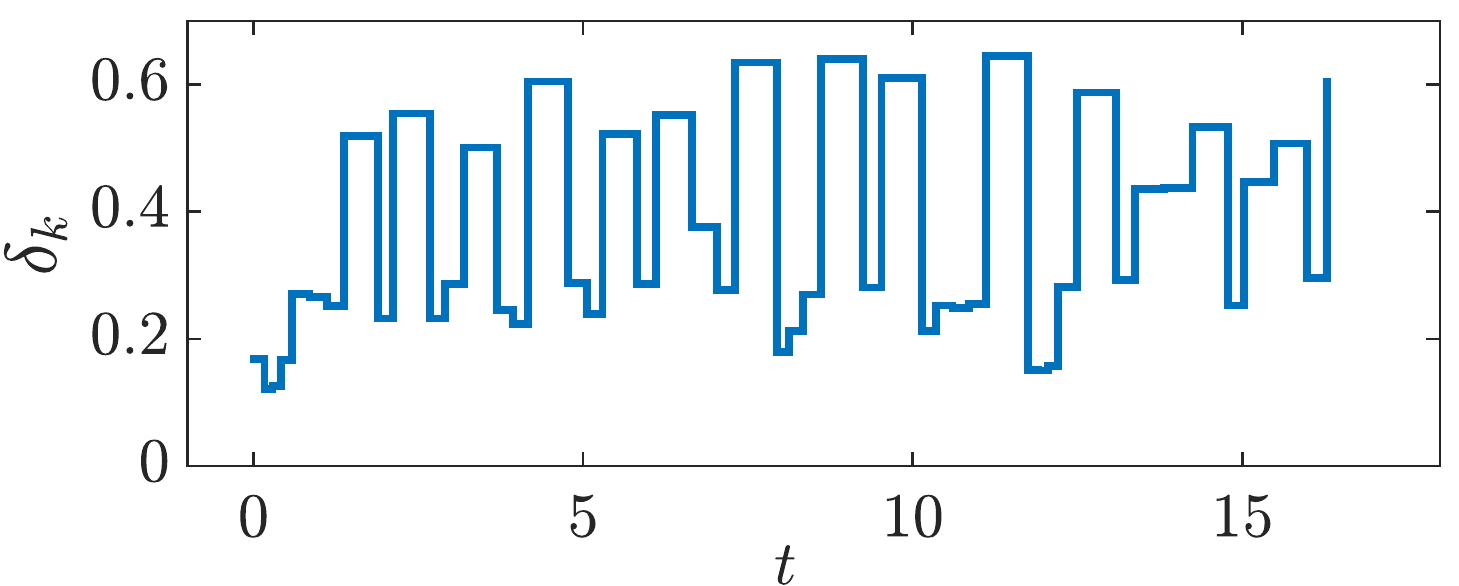}
		}
		\caption{\small Inter-execution times $\delta_k$ versus time $t$.}
		\label{figvs3}
	\end{figure}

	
	\begin{table}[t] 
		\caption{\small Dependency of quantities $R_F$, $R_u$, and $D$ on parameter $\alpha$.} 
		\centering
		
		\begin{tabular}{|c| c |c|c|} 
			\hline
			$\alpha$ & $R_F$ & $R_u$ & $D$  \\ 
			\hline
			$1.05$ & $85.56 \%$ & $62.54 \%$ & $0.2823 $ \\ 
			\hline
			$1.10$ & $72.60 \%$ & $66.89 \%$ & $0.2715 $ \\ 
			\hline
			$1.15$ & $61.61 \%$ & $49.67 \%$ & $0.3322 $ \\
			\hline
			$1.20$ & $53.86 \%$ & $41.04 \%$ & $0.3247 $ \\
			\hline
			$1.25$ & $50.63 \%$ & $36.81 \%$ & $0.3629 $ \\
			\hline
			$1.30$ & $49.05 \%$ & $32.00 \%$ & $0.3529 $ \\
			\hline
		\end{tabular} \label{TDep}
	\end{table} 
	
	\begin{figure}[t]
		\centering
		
		\subfloat[\label{FigRF}]{
			\includegraphics[width=6.4cm]{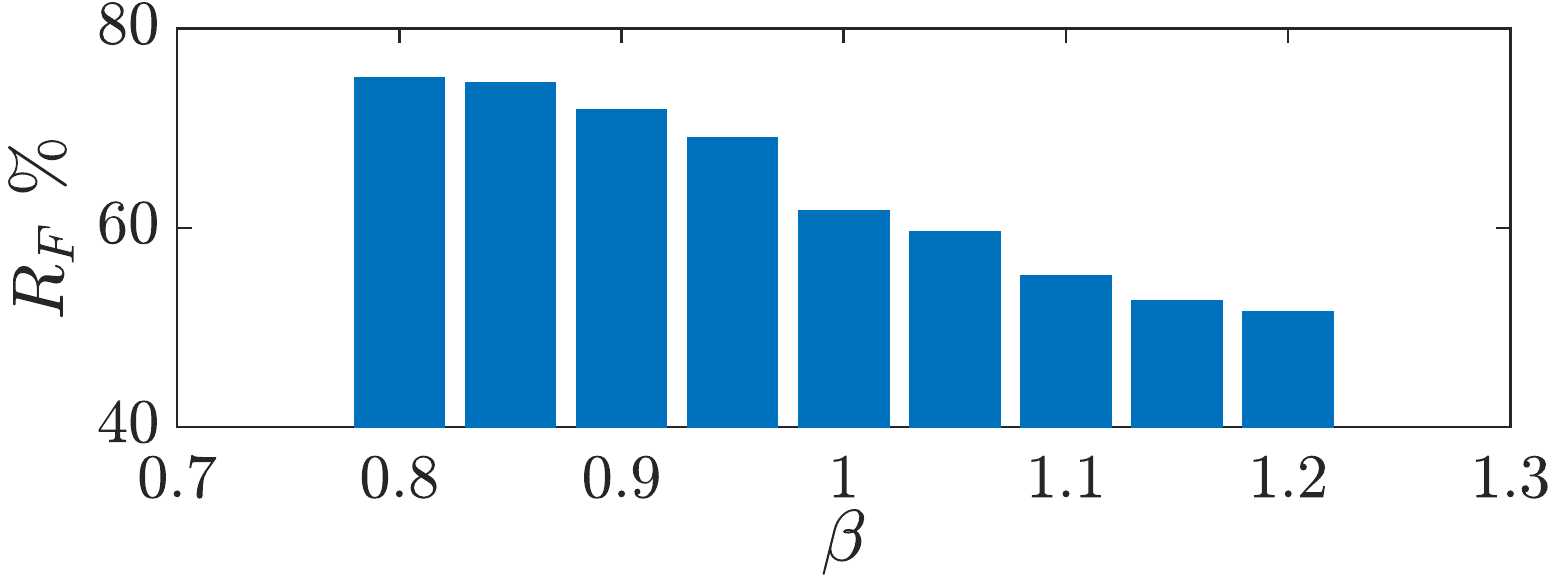}
		}
		\hfill
		\subfloat[\label{FigRU}]{
			\includegraphics[width=6.4cm]{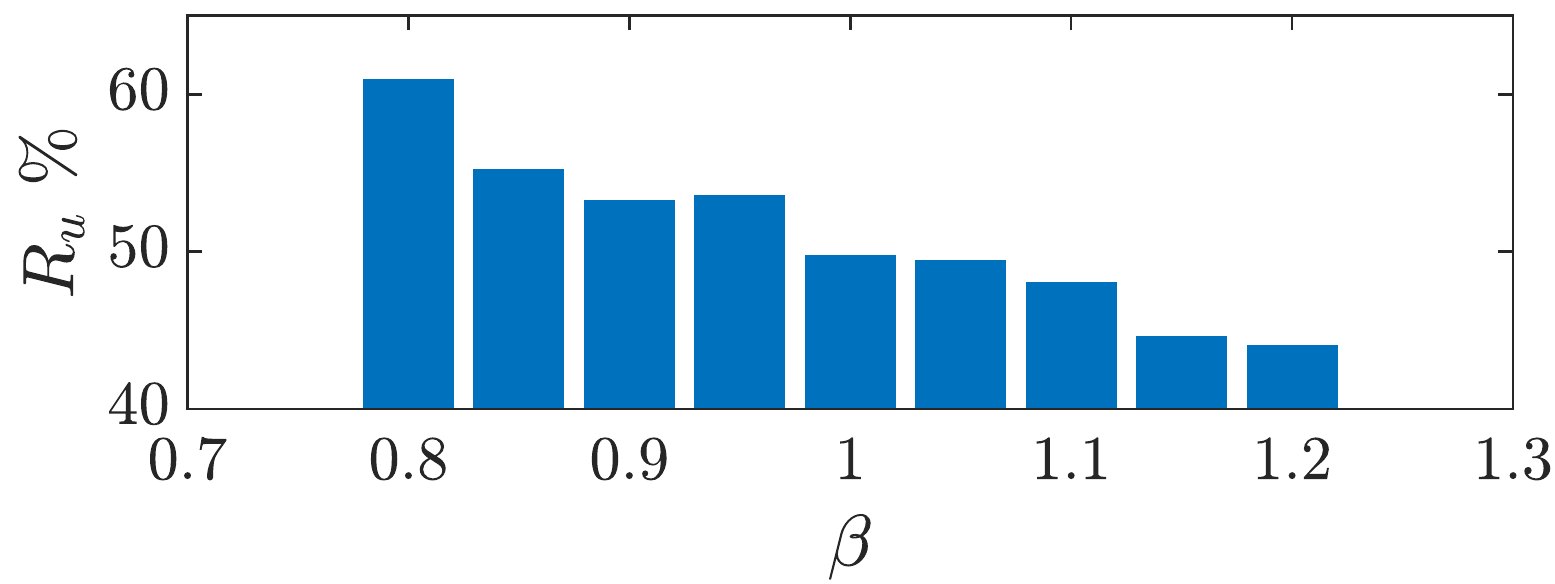}
		}
		\caption{\small (a)   ${R}_F$ versus decay rate $\beta$. (b)  ${R}_u$ versus   decay rate  $\beta$.}
	\end{figure}
	
	\begin{figure}[t]
		\centering
		
		{
			\includegraphics[width=6.6cm]{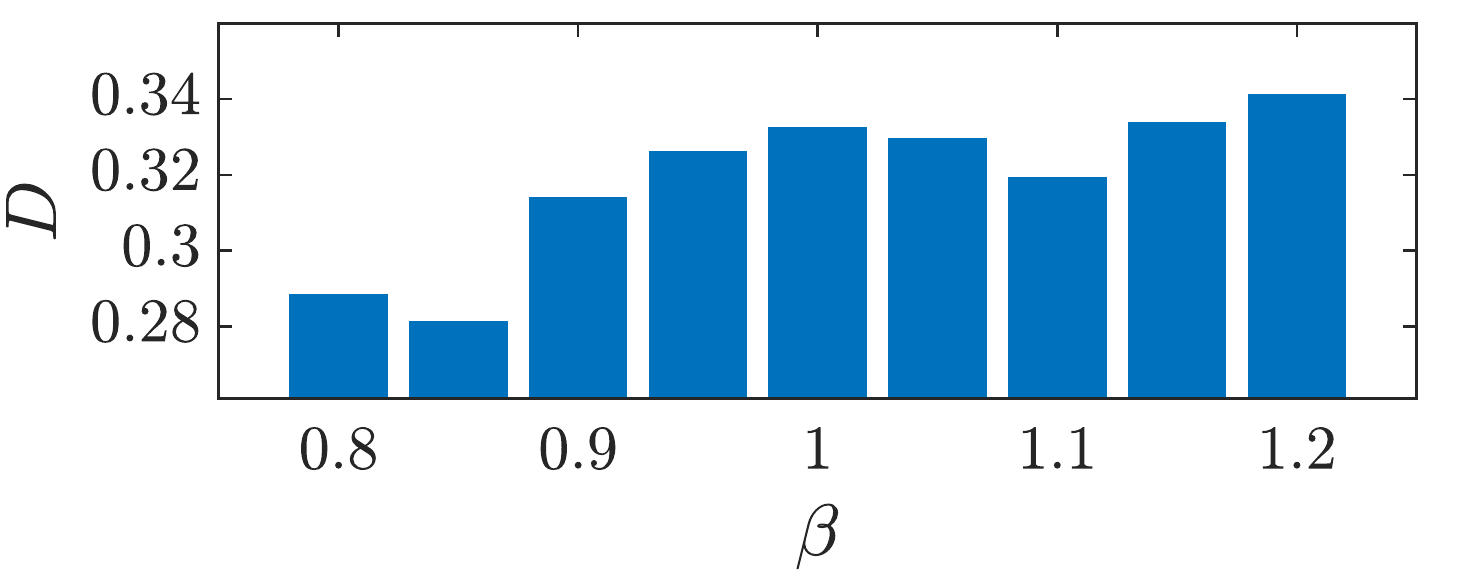}
		}
		\caption{\small Average inter-execution time $D$ versus spatial decay rate $\beta$.}
		\label{FigD}
	\end{figure}
	
	We consider the previous setup and study the effect of penalizing parameters $\gamma$ and $\eta$. To investigate the effect of  $\gamma$ on ${R}_F$, we set $\alpha = 1.15$ (i.e., accepting at most $15\%$ performance loss), $\eta=0.001$, and vary $\gamma$ in  $\left [10^{-5},10^{-3}\right ]$. The result of Algorithm $1$ shows that ${R}_F$ decreases as $\gamma$ grows (see Fig. \ref{figvs4}), which implies that the average sensing requirements of the subsystems from each other has decreased.  
	We repeat the previous study for fixed value of $\gamma=10^{-3}$ and by varying $\eta$ in $\left [10^{-5},10^{-3}\right ]$. Fig. \ref{figvs5} illustrates the relationship between $\eta$ and ${R}_u$. Similar to $({R}_F$,$\gamma)$ relationship, there exists a tradeoff between ${R}_u$ and $\eta$: as  $\eta$ increases, on average, it enforces more components of $u_k$ to be zero; i.e.,  less number of actuators on average are utilized (in time). 
	
	\begin{figure}[t]
		\centering
		\vspace{.2cm}
		
		\subfloat[\label{figvs4}]{
			\includegraphics[width=6.7cm]{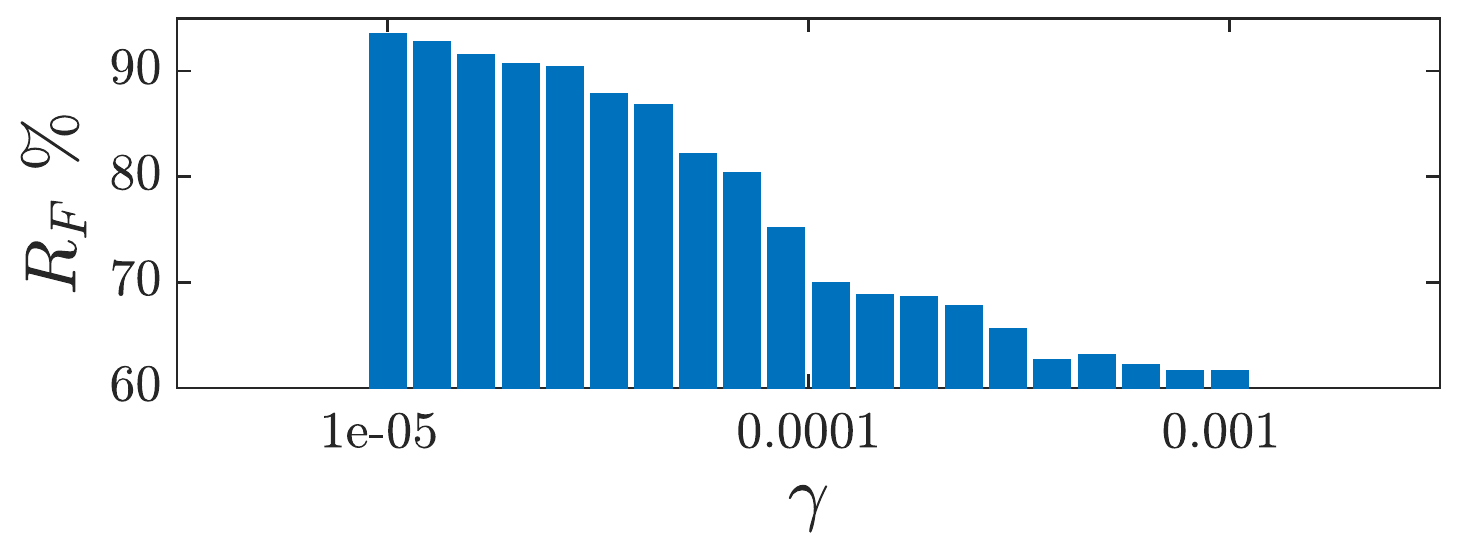}
		}         
		\hfill
		\subfloat[\label{figvs5}]    {
			\includegraphics[width=6.7cm]{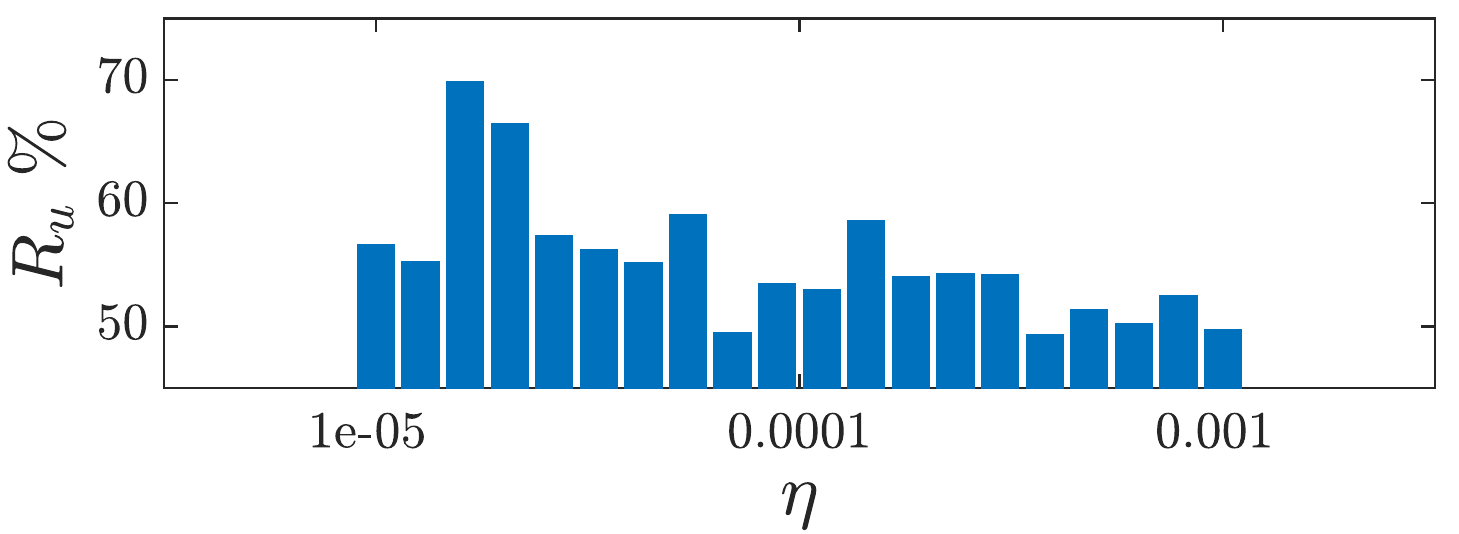}
		}      
		\caption{\small (a)   ${R}_F$ versus penalizing parameter $\gamma$. (b)  ${R}_u$ versus penalizing parameter $\eta$.}
	\end{figure}

%
%
%

\section{Discussion and Conclusion} \label{CoDi}
We present a control algorithm based on blending ideas from self-triggered control and $\ell_1$-regularized optimal control. The cornerstones of our algorithm are: estimation of the inter-execution time while feedback gain is kept fixed (via a nonlinear optimization), design of the sparse optimal controller while inter-execution time is kept fixed (via a convex optimization), and correcting the value of the next triggering times using the resulting feedback gain. The closed-loop stability is guaranteed by imposing a performance-ensuring constraint. Our extensive  numerical simulations assert that our proposed algorithm improves space-time sparsity in sensing, communication, and actuation in an interconnected network of multiple systems.   In the case of spatially decaying systems, our simulations show that there exists a tradeoff between the spatial decay rate and sparsity of the communication graph (i.e., sparsity of the feedback gains). Likewise, as  the spatial decay rate increases, the average sampling rate decreases. One of our future research directions is to modify our algorithm and implement it   in a distributed and asynchronous manner using localized information. 

\section*{Acknowledgment}
The authors would like to thank Prof. Paulo Tabuada for his fruitful comments and discussions. 

\bibliographystyle{plain}
\bibliography{autosam}

\end{document}